\documentclass{scrartcl} 
\usepackage[a4paper,hmargin=1.5cm,vmargin=3 cm]{geometry}
\usepackage{amsmath,amssymb}

\usepackage{url}
\usepackage{dsfont}
\usepackage{mathtools}
\usepackage{algorithm}
\usepackage{algpseudocode}
\usepackage{graphicx}
\usepackage{amsthm}
\usepackage{enumerate}
\usepackage{enumitem}
\setenumerate{leftmargin=6mm}
\setitemize{leftmargin=4mm}

\usepackage{subfigure}
\usepackage{color}
\definecolor{green}{rgb}{0.2, 0.6, 0.2}
\definecolor{lightgreen}{rgb}{0.5, 0.9, 0.5}
\definecolor{lightred}{rgb}{1, 0.7, 0.7}
\definecolor{gray}{rgb}{0.87, 0.87, 0.87}

\usepackage{hyperref}
\hypersetup{
    bookmarks=true,         
    pdfauthor={Marc Peter Deisenroth and Henrik Ohlsson},     
    colorlinks,%
    citecolor=black,%
    filecolor=black,%
    linkcolor=black,%
    urlcolor=black
}
\usepackage[all]{hypcap}

\newcommand{\R}[0]{\mathds{R}} 
\newcommand{\gauss}[2]{\mathcal{N}(#1,#2)}
\newcommand{\gaussBig}[2]{\mathcal{N}\left(#1,#2\right)}
\newcommand{\gaussx}[3]{\mathcal{N}(#1\,|\,#2,#3)}

\renewcommand{\d}{\operatorname{d}\!}
\newcommand{\E}{\mathds{E}} 
\newcommand{\cov}[0]{\mathrm{cov}} 
\renewcommand{\vec}[1]{{\boldsymbol{\mathbf{#1}}}} 
\newcommand{\mat}[1]{{\ensuremath{\mathbf{#1}}}} 
\newcommand{\inv}[0]{^{-1}} 
\newcommand{\T}[0]{^\top} 
\newcommand{\prob}{{p}} 

\newcommand{\obs}[0]{z}      

\newcommand{\meas}[1]{{#1}}

\newcommand{\idx}[1]{^{(#1)}}

\newcommand{\green}[1]{{\bf{\textcolor{green}{#1}}}}
\newcommand{\red}[1]{{\bf{\textcolor{red}{#1}}}}

\newtheorem{proposition}{Proposition}

\title{A Probabilistic Perspective on Gaussian Filtering and
  Smoothing\footnote{This paper is an extended version of the
    conference paper~\cite{Deisenroth2011}.}}
\author{Marc Peter Deisenroth$^{1,2}$ and Henrik Ohlsson$^3$\\[5mm]
  \small{$^1$Department of Computer Science \& Engineering}\\[-2mm]
  \small{University of Washington, Seattle, USA}\\[2mm]
  \small{$^2$Department of Engineering}\\[-2mm]
  \small{University of Cambridge, UK}\\[2mm]
  \small{$^3$Department of Electrical Engineering}\\[-2mm]
  \small{Link\"oping University, Sweden}}

\begin{document}

\maketitle

\begin{abstract}
 We present a general probabilistic perspective on Gaussian filtering
  and smoothing. This allows us to show that common approaches to
  Gaussian filtering\slash smoothing can be distinguished solely by
  their methods of computing\slash approximating the means and
  covariances of joint probabilities. This implies that novel filters
  and smoothers can be derived straightforwardly by providing methods
  for computing these moments. Based on this insight, we derive the
  cubature Kalman smoother and propose a novel robust filtering and
  smoothing algorithm based on Gibbs sampling.
\end{abstract}

Inference in latent variable models is about extracting
information about a not directly observable quantity, the latent
variable, from noisy observations. Both recursive and batch methods
are of interest and referred to as \emph{filtering} respective
\emph{smoothing}. Filtering and smoothing in latent variable time
series models, including hidden Markov models and dynamic systems,
have been playing an important role in signal processing, control, and
machine learning for decades~\cite{Kalman1960,Maybeck1979,Bishop2006}.

In the context of dynamic systems, filtering is widely used in control
and robotics for online Bayesian state estimation~\cite{Thrun2005},
while smoothing is commonly used in machine learning algorithms for
parameter learning~\cite{Bishop2006}. For computational efficiency
reasons, many filters and smoothers approximate appearing probability
distributions by Gaussians. This is why they are referred to as
\emph{Gaussian filters/smoothers}.

In the following, we discuss Gaussian filtering and smoothing from a
general probabilistic perspective without focusing on particular
implementations. We identify the high-level concepts and the
components required for filtering and smoothing, while avoiding
getting lost in the implementation and computational details of
particular algorithms (see e.g., standard derivations of the Kalman
filter~\cite{Anderson2005,Thrun2005}).

We show that for Gaussian filters\slash smoothers for (non)linear
systems (including common algorithms such as the extended Kalman
filter (EKF)~\cite{Maybeck1979}, the cubature Kalman filter
(CKF)~\cite{Arasaratnam2009}, or the unscented Kalman filter
(UKF)~\cite{Julier2004}) can be distinguished by their means to
computing Gaussian approximations to the joint probability
distributions $\prob(\vec x_{t-1},\vec x_t|\vec\obs_{1:t-1})$ and
$\prob(\vec x_t,\vec\obs_t|\vec\obs_{1:t-1})$. 
Our results also imply that novel filtering and
smoothing algorithms can be derived straightforwardly, given a method
to determining the moments of these joint distributions. Using this
insight, we present and analyze the cubature Kalman smoother (CKS) and
a filter and an RTS smoother based on Gibbs sampling.

We start this paper by setting up the problem and the notation,
Sec.~\ref{sec:probform}.  We thereafter proceed by reviewing Gaussian
filtering and RTS smoothing from a high-level probabilistic
perspective to derive sufficient conditions for Gaussian filtering and
smoothing, respectively (Secs.~\ref{sec:filtering}
and~\ref{sec:smoothing}). The implications of this result are
discussed in Sec.~\ref{sec:results}, which lead to the derivation of a
novel Gaussian filter and RTS smoother based on Gibbs sampling.
Sec.~\ref{sec:numerical evaluation} provides proof-of-concept
numerical evaluations for the proposed method for both linear and
nonlinear systems. Secs.~\ref{sec:discussion}--\ref{sec:conclusion}
discuss related work and conclude the paper.

\section{Problem Setup and Notation}
\label{sec:probform}

We consider discrete-time stochastic dynamic systems of the form
\begin{align}
\vec x_t &= f(\vec x_{t-1}) + \vec w_{t}\,,\label{eq:system equation}\\
\vec\obs_t& = g(\vec x_t) + \vec v_t\,,\label{eq:measurement}
\end{align}
where $\vec x_t\in\R^D$ is the state, $\vec\obs_t\in\R^E$ is the
measurement at time step $t=1,\dotsc,T$, $\vec w_t\sim\gauss{\vec
  0}{\mat Q}$ is i.i.d. Gaussian system noise, $\vec
v_t\sim\gauss{\vec 0}{\mat R}$ is i.i.d. Gaussian measurement noise,
$f$ is the transition\slash system function and $g$ is the measurement
function.  The graphical model of the considered dynamic system is
given in fig.~\ref{fig:gm}.
\begin{figure}[tb]
\centering
\includegraphics[width = 0.7\hsize]{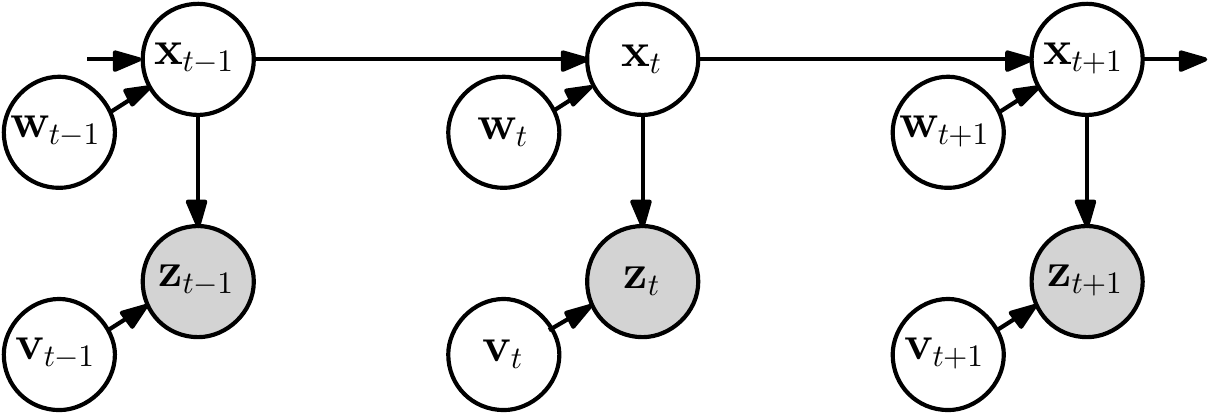}
\caption{Graphical model of the dynamic system. The shaded nodes are
  the measured variables $\vec\obs_t$, the unshaded nodes are
  unobserved variables. The arrows represent probabilistic
  dependencies between the variables.}
\label{fig:gm}
\end{figure}

The noise covariance matrices $\mat Q$, $\mat R$, the system function
$f$, and the measurement function $g$ are assumed known. If not stated
otherwise, we assume nonlinear functions $f$ and $g$. The initial
state $\vec x_0$ of the time series is distributed according to a
Gaussian prior distribution $\prob(\vec x_0) =
\gauss{\vec\mu_0^x}{\mat\Sigma_0^x}$.  The purpose of filtering and
smoothing is to find approximations to the posterior distributions
$\prob(\vec x_t|\vec\obs_{1:\tau})$, where a subscript $1\!:\!\tau$
abbreviates $1,\dotsc,\tau$, with $\tau\!=\!t$ for filtering and
$\tau\!=\!T$ for smoothing.

In this paper, we consider Gaussian approximations $\gaussx{\vec
  x_t}{\vec\mu_{t|\tau}^x}{\mat\Sigma_{t|\tau}^x}$ of the latent state
posteriors $\prob(\vec x_t|\vec\obs_{1:\tau})$. We use
the shorthand notation $\vec a_{b|c}^d$ where $\vec a=\vec\mu$ denotes
the mean $\vec\mu$ and $ \vec a = \mat\Sigma$ denotes the covariance,
$b$ denotes the time step under consideration, $c$ denotes the time
step up to which we consider measurements, and $d\in\{x,\obs\}$
denotes either the latent space ($x$) or the observed space ($\obs$).

Let us assume a prior $\prob(\vec x_0) = \prob(\vec x_0|\emptyset)$
and a sequence $\meas{\vec\obs}_1,\dotsc,\vec\obs_T$ of noisy
measurements of the latent states $\vec x_0,\dotsc,\vec x_T$ through
the measurement function $g$. The objective of \emph{filtering} is to
compute a posterior distribution $\prob(\vec
x_t|\meas{\vec\obs}_{1:t})$ over the latent state as soon as a new
measurement $\vec\obs_t$ is available. \emph{Smoothing} extends
filtering and aims to compute the posterior state distribution of the
hidden states $\vec x_t$, $t=0,\dotsc,T$, given \emph{all}
measurements $\vec\obs_1,\dotsc,\vec\obs_T$ (see
e.g.,~\cite{Anderson2005,Thrun2005}).

\section{Gaussian Filtering}
\label{sec:filtering}
Given a prior $\prob(\vec x_0)$ on the initial state and a dynamic
system (e.g., Eqs.~(\ref{eq:system
  equation})--(\ref{eq:measurement})), the objective of
\emph{filtering} is to infer a posterior distribution $\prob(\vec
x_t|\vec\obs_{1:t})$ of the hidden state $\vec x_t$, $t = 1,\dotsc,T$,
incorporating the evidence of the measurements
$\vec\obs_{1:t}$. Specific for Gaussian filtering is that posterior
distributions are approximated by
Gaussians~\cite{Thrun2005}. Approximations are required since
generally a Gaussian distribution mapped through a nonlinear function
does not stay Gaussian.

Assume a Gaussian filter distribution $\prob(\vec
x_{t-1}|{\vec\obs}_{1:t-1})=\gauss{\vec\mu_{t-1|t-1}^x}{\mat\Sigma_{t-1|t-1}^x}$
is given (if not, we employ the prior $\prob(\vec x_0) = \prob(\vec
x_0|\emptyset)=\gauss{\vec\mu_{0|\emptyset}^x}{\mat\Sigma_{0|\emptyset}^x})$
on the initial state.  Using Bayes' theorem, the filter
distribution at time $t$ is
\begin{align}
  \prob(\vec x_t|{\vec\obs}_{1:t}) &= \frac{\prob(\vec
    x_t,\vec\obs_t|\vec\obs_{1:t-1})}{\prob(\vec\obs_t|
    \vec\obs_{1:t-1})}\propto
  \prob(\vec\obs_t|\vec x_t)\prob(\vec x_t|\vec\obs_{1:t-1})\,. 
\label{eq:Bayes theorem}
\end{align}

%
\begin{proposition}[Filter Distribution]\label{prop:filter
    distribution}
  Gaussian filters approximate the filter
  distribution $\prob(\vec
  x_t|\vec\obs_{1:t})$ using a Gaussian distribution
  $\gauss{\vec\mu_{t|t}^x}{\mat\Sigma_{t|t}^x}$. The moments of this
  approximation are in general computed through
\begin{align}
  \vec\mu_{t|t}^x&=\hat{\vec\mu}_{t|t-1}^x +
  \hat{ \mat\Sigma}_{t|t-1}^{x\obs}(\hat{\mat\Sigma}_{t|t-1}^\obs)\inv(\vec\obs_t -
  \hat{  \vec\mu}_{t|t-1}^\obs)\,,
  \label{eq:proposition filter mean}\\
  \mat\Sigma_{t|t}^x&= \hat{\mat\Sigma}_{t|t-1}^{x} - \hat{
    \mat\Sigma}_{t|t-1}^{x\obs}(\hat{
    \mat\Sigma}_{t|t-1}^\obs)\inv\hat{ \mat\Sigma}_{t|t-1}^{\obs x}\,.
\label{eq:proposition filter covariance}
\end{align}
Since the
true moments of the joint distribution  $\prob(\vec
x_t,\vec\obs_t|\vec \obs_{1:t-1})$ can in general not be computed
analytically, approximations\slash estimates  are used
(hence the $\hat{~}$-symbols).
\end{proposition}

\begin{proof}
  Generally, filtering proceeds by alternating between predicting
  (\emph{time update}) and correcting (\emph{measurement
    update})~\cite{Anderson2005,Thrun2005}:
\begin{enumerate}
\item Time update (predictor)
\begin{enumerate}
\item Compute the predictive distribution $\prob(\vec
  x_t|\vec\obs_{1:t-1})$.
\end{enumerate}
\item Measurement update (corrector) 
\begin{enumerate}
\item Compute the joint distribution $\prob(\vec
  x_t,\vec\obs_t|{\vec\obs}_{1:t-1})$ of the next latent state and the
  next measurement.
\item Measure $\vec\obs_t$.
\item Compute the posterior $\prob(\vec x_t|{\vec\obs}_{1:t})$.
\end{enumerate}
\end{enumerate}
In the following, we detail these steps to prove
Prop.~\ref{prop:filter distribution}.

\subsection{Time Update (Predictor)}
\begin{enumerate}
\item[(a)] Compute the predictive distribution $\prob(\vec
  x_t|\vec\obs_{1:t-1})$. The predictive distribution of state $\vec
  x$ at time $t$ given the evidence of measurements up to time $t-1$
  is 
\begin{align}
\prob(\vec x_t|\vec\obs_{1:t-1}) \!= \! \int\!\prob(\vec x_t|\vec
x_{t-1})\prob(\vec x_{t-1}|\vec\obs_{1:t-1}) \d\vec
x_{t-1}\,,\label{eq:time update}
\end{align}
where $\prob(\vec x_t|\vec x_{t-1}) = \gaussx{\vec x_t}{f(\vec
  x_{t-1})}{\mat Q}$ is the transition probability. In Gaussian
filters, the predictive distribution $\prob(\vec
x_t|\vec\obs_{1:t-1})$ in Eq.~\eqref{eq:time update} is approximated
by a Gaussian distribution, whose exact mean and covariance are given by
\begin{align}
  \vec\mu_{t|t-1}^x&\!\coloneqq\!\E_{\vec x_t}[\vec
  x_t|\vec\obs_{1:t-1}]\! =\! \E_{\vec x_{t-1},\vec w_t}[f(\vec
  x_{t-1})\!+\!\vec w_t|\vec\obs_{1:t-1}] =\int f(\vec
  x_{t-1})\prob(\vec x_{t-1}|\vec\obs_{1:t-1})\d\vec x_{t-1}\,,
  \label{eq:mean state predicted measurement}\\
  \mat\Sigma_{t|t-1}^x &\coloneqq\cov_{\vec x_t}[\vec
  x_t|\vec\obs_{1:t-1}] = \cov_{\vec x_{t-1}}[f(\vec
  x_{t-1})|\vec\obs_{1:t-1}] + \cov_{\vec w_t}[\vec w_t] \nonumber \\
&=\underbrace{\int f(\vec x_{t-1})f(\vec x_{t-1})\T\prob(\vec
    x_{t-1}|\vec\obs_{1:t-1})\d\vec
    x_{t-1}-\vec\mu_{t|t-1}^x\big(\vec\mu_{t|t-1}^x)\T}_{ =
    \cov_{\vec x_{t-1}}[f(\vec x_{t-1})|\vec\obs_{1:t-1}]} +\underbrace{\mat
    Q}_{ =\cov_{\vec w_t}[\vec w_t]}
  \label{eq:cov state predicted measurement}
\end{align}
respectively. In Eq.~\eqref{eq:mean state predicted measurement}, we
exploited that the noise term $\vec w_t$ in Eq.~\eqref{eq:system
  equation} has mean zero and is independent.  A Gaussian
approximation to the time update $\prob(\vec x_t|\vec\obs_{1:t-1})$ is
then given by $\gaussx{\vec
  x_t}{\vec\mu_{t|t-1}^x}{\mat\Sigma_{t|t-1}^x}$.
\end{enumerate}

\subsection{Measurement Update (Corrector)}
\begin{enumerate}
\item[(a)] Compute the joint distribution 
\begin{align}
\prob(\vec
  x_t,\vec\obs_t|\vec\obs_{1:t-1})=\prob(\vec\obs_t|\vec
  x_t)\underbrace{\prob(\vec x_t|\vec\obs_{1:t-1})}_{\text{time update}}\,.
\label{eq:1st joint}
\end{align}
In Gaussian filters, a Gaussian approximation to this joint is an
intermediate step toward the desired Gaussian approximation of the
posterior $\prob(\vec x_t|\vec\obs_{1:t})$. If the mean and the
covariance of the joint in Eq.~(\ref{eq:1st joint}) can be computed or
estimated, the desired filter distribution corresponds to the
conditional $\prob(\vec x_t|\vec\obs_{1:t})$ and is given in closed
form~\cite{Bishop2006}.

Our objective is to compute a Gaussian approximation
\begin{align}
\gaussBig{
\begin{bmatrix}
\vec\mu_{t|t-1}^x\\
\vec\mu_{t|t-1}^\obs
\end{bmatrix}
}{
\begin{bmatrix}
\mat\Sigma_{t|t-1}^x & \mat\Sigma_{t|t-1}^{x\obs}\\
\mat\Sigma_{t|t-1}^{\obs x} & \mat\Sigma_{t|t-1}^\obs
\end{bmatrix}
}
\label{eq:joint p(x,z)}
\end{align}
to the joint $\prob(\vec x_t,\vec\obs_t|\vec\obs_{1:t-1})$ in
Eq.~(\ref{eq:1st joint}). Since a Gaussian approximation
$\gauss{\vec\mu_{t|t-1}^x}{\mat\Sigma_{t|t-1}^x}$ to the marginal
$\prob(\vec x_t|\vec\obs_{1:t-1})$ is known from the time update, 
it remains to compute the
marginal $\prob(\vec\obs_t|\vec\obs_{1:t-1})$ and the
cross-covariance $\mat\Sigma_{t|t-1}^{x\obs}\coloneqq \cov_{\vec
  x_t,\vec\obs_t}[\vec x_t,\vec\obs_t|\vec\obs_{1:t-1}]$.
\begin{itemize}
\item The marginal $\prob(\vec
  z_t|{\vec\obs}_{1:t-1})$ of the joint in Eq.~\eqref{eq:joint
    p(x,z)} is
\begin{align}
\prob(\vec\obs_t|\vec\obs_{1:t-1}) =  \int\prob(\vec\obs_t|\vec
x_t)\prob(\vec x_t|\vec\obs_{1:t-1}) \d\vec x_t\,,\nonumber
\end{align}
where the state $\vec x_t$ is integrated out according to the time
update $\prob(\vec x_t|\vec\obs_{1:t-1})$. 
The measurement Eq.~(\ref{eq:measurement}),
yields $\prob(\vec\obs_t|\vec x_t)=\gauss{g(\vec x_t)}{\mat R}$.
Hence, the exact mean of the marginal is
\begin{align}
  \vec\mu_{t|t-1}^\obs&\coloneqq
  \E_{\vec\obs_t}[\vec\obs_t|\vec\obs_{1:t-1}]= \E_{\vec x_t}[g(\vec
  x_t)|\vec\obs_{1:t-1}] =\int g(\vec x_t)\underbrace{\prob(\vec
  x_t|\vec\obs_{1:t-1})}_{\text{time update}}\d\vec x_t
\label{eq:mean predicted measurement}
\end{align}
since the noise term $\vec v_t$ in the measurement
Eq.~(\ref{eq:measurement}) is independent and has zero
mean. Similarly, the exact covariance of the marginal
$\prob(\vec\obs_t|\vec\obs_{1:t-1})$ is
\begin{align}
  \mat\Sigma_{t|t-1}^\obs &=
  \cov_{\vec\obs_t}[\vec\obs_t|\vec\obs_{1:t-1} = \cov_{\vec
    x_t}[g(\vec x_t)|\vec\obs_{1:t-1}]+ \cov_{\vec
    v_t}[\vec v_t]\nonumber \\
 & = \underbrace{\int g(\vec x_t)g(\vec x_t)\T\prob(\vec
    x_t|\vec\obs_{1:t-1})\d\vec
    x_t-\vec\mu_{t|t-1}^\obs\big(\vec\mu_{t|t-1}^\obs)\T}_{=\cov_{\vec
    x_t}[g(\vec
    x_t)|\vec\obs_{1:t-1}]} +\underbrace{\mat
    R}_{=\cov_{\vec v_t}[\vec v_t]}\,.
\label{eq:covariance predicted measurement}
\end{align}
Hence, a Gaussian approximation to the marginal measurement
distribution $\prob(\vec\obs_t|\vec\obs_{1:t-1})$ is given by
\begin{align}
    \gaussx{\vec\obs_t}{\vec\mu_{t|t-1}^\obs}{\mat\Sigma_{t|t-1}^\obs}\,,
\label{eq:measurement marginal}
\end{align}
with the mean and covariance given in Eqs.~(\ref{eq:mean
  predicted measurement}) and~(\ref{eq:covariance predicted
  measurement}), respectively.

\item Due to the independence of $\vec v_t$, the exact
  cross-covariance terms of the joint in Eq.~\eqref{eq:joint p(x,z)}
  are
\begin{align}
  \mat\Sigma_{t|t-1}^{x\obs}&=\cov_{\vec x_t,\vec\obs_t}[\vec
  x_t,\vec\obs_t|\vec\obs_{1:t-1}] =\E_{\vec x_t,\vec z_t}[\vec
  x_t\vec\obs_t\T|\vec\obs_{1:t-1}]-\E_{\vec x_t}[\vec
   x_t|\vec\obs_{1:t-1}]\E_{\vec\obs_t}[\vec\obs_t|\vec\obs_{1:t-1}]\T\nonumber\\
  &=\iint\vec x_t\vec\obs_t\T\prob(\vec x_t,\vec\obs_t|\vec\obs_{1:t-1})\d\vec
  \obs_t\d\vec x_t-\vec\mu_{t|t-1}^x(\vec\mu_{t|t-1}^\obs)\T\nonumber\,.
\end{align}
Plugging in the measurement Eq.~\eqref{eq:measurement}, we obtain
\begin{align}
  \mat\Sigma_{t|t-1}^{x\obs} & = \int\vec x_tg(\vec x_t)\T\prob(\vec
  x_t|\vec\obs_{1:t-1})\d\vec
  x_t-\vec\mu_{t|t-1}^x(\vec\mu_{t|t-1}^\obs)\T\,.
  \label{eq:cross-covariance x and z}
\end{align}
\end{itemize}
\item[(b)] Measure $\vec\obs_t$.
\item[(c)] Compute a Gaussian approximation of the posterior
  $\prob(\vec x_t|\vec\obs_{1:t})$. This boils down to computing a
  conditional from the Gaussian approximation to the joint
  distribution $\prob(\vec x_t,\vec\obs_t|\vec\obs_{1:t-1})$ in
  Eq.~\eqref{eq:joint p(x,z)}. The expressions from Eqs.~(\ref{eq:mean
    state predicted measurement}),~\eqref{eq:cov state predicted
    measurement},~(\ref{eq:mean predicted
    measurement}),~(\ref{eq:covariance predicted measurement}),
  and~(\ref{eq:cross-covariance x and z}), yield a Gaussian
  approximation $\gaussx{\vec
    x_t}{\vec\mu_{t|t}^x}{\mat\Sigma_{t|t}^x}$ of the filter
  distribution $\prob(\vec x_t|\vec\obs_{1:t})$, where
\begin{align}
 \vec\mu_{t|t}^x &= \vec\mu_{t|t-1}^x +
  \mat\Sigma_{t|t-1}^{x\obs}\big(\mat\Sigma_{t|t-1}^\obs\big)\inv
  (\vec\obs_t-\vec\mu_{t|t-1}^\obs)\,,\label{eq:generic filter mean}\\     
 \mat\Sigma_{t|t}^x &= \mat\Sigma_{t|t-1}^x -
  \mat\Sigma_{t|t-1}^{x\obs}\big(\mat\Sigma_{t|t-1}^\obs\big)\inv
  \mat\Sigma_{t|t-1}^{\obs x}\,.
\label{eq:generic filter covariance}
\end{align}

\end{enumerate}
Generally, the required integrals in Eqs.~\eqref{eq:mean state
  predicted measurement}, \eqref{eq:cov state predicted measurement},
\eqref{eq:mean predicted measurement}, \eqref{eq:covariance predicted
  measurement}, and \eqref{eq:cross-covariance x and z} cannot be
computed analytically. Hence, approximations of the moments are
typically used in Eqs.~\eqref{eq:generic filter mean} and
\eqref{eq:generic filter covariance}. This concludes the proof of
Prop.~\ref{prop:filter distribution}.
\end{proof}


\subsection{Sufficient Conditions for Gaussian Filtering}

In any Bayes filter~\cite{Thrun2005}, the sufficient components to
computing the Gaussian filter distribution in Eqs.~(\ref{eq:generic
  filter mean}) and~(\ref{eq:generic filter covariance}) are the mean
and the covariance of the joint distribution $\prob(\vec
x_t,\vec\obs_t|\vec \obs_{1:t-1})$.
Generally, the required integrals in Eqs.~\eqref{eq:mean state
  predicted measurement}, \eqref{eq:cov state predicted measurement},
\eqref{eq:mean predicted measurement}, \eqref{eq:covariance predicted
  measurement}, and \eqref{eq:cross-covariance x and z} cannot be
computed analytically. One exception are \emph{linear} functions $f$
and $g$, where the analytic solutions to the integrals are embodied in
the Kalman filter~\cite{Kalman1960}: Using the rules
of predicting in linear Gaussian systems, the Kalman filter equations
can be recovered when plugging in the respective means and covariances
into Eq.~(\ref{eq:generic filter mean}) and~(\ref{eq:generic filter
  covariance})~\cite{Roweis1999, Minka1998, Anderson2005, Bishop2006,
  Thrun2005}. 
In many \emph{nonlinear} dynamic systems, filtering algorithms
approximate probability distributions (see e.g., the
UKF~\cite{Julier2004} and the CKF~\cite{Arasaratnam2009}) or the
functions $f$ and $g$ (see e.g., the EKF~\cite{Maybeck1979} or the
GP-Bayes filters~\cite{Deisenroth2009a, Ko2009}). Using the means and
(cross-)covariances computed by these algorithms and plugging them
into Eqs.~(\ref{eq:generic filter mean})--(\ref{eq:generic filter
  covariance}), recovers the corresponding filter update equations for
the EKF, the UKF, the CKF, and the GP-Bayes filters.


\section{Gaussian RTS Smoothing}
\label{sec:smoothing}
In this section, we present a general probabilistic perspective on
Gaussian RTS smoothers and derive sufficient conditions for Gaussian
smoothing.

The smoothed state distribution is the posterior distribution of the
hidden state given \emph{all} measurements
\begin{equation}
\prob(\vec x_t|\vec\obs_{1:T})\,,\,\,t=T,\dotsc,0\,.
\end{equation}

\begin{proposition}[Smoothing Distribution]\label{prop:smoothing}
  For Gaussian smoothers, the mean and the covariance of a Gaussian
  approximation to the distribution $\prob(\vec x_t|\vec\obs_{1:T})$
  are generally computed as
\begin{align}
  \vec\mu_{t-1|T}^x&=\hat{\vec\mu}_{t-1|t-1}^x + \hat{\mat
    J}_{t-1}(\hat{\vec\mu}_{t|T}^x - \hat{\vec\mu}_{t|t-1}^x)\,,
  \label{eq:proposition mean smoother}\\
  \mat \Sigma_{t-1|T}^x &=\hat{\mat \Sigma}_{t-1|t-1}^x + \hat{\mat
    J}_{t-1} (\hat{\mat \Sigma}_{t|T}^x - \hat{\mat
    \Sigma}_{t|t-1}^x)\hat{\mat J}_{t-1}\T\,,
  \label{eq:proposition cov smoother}\\
  \mat J_{t-1} &= \cov[\vec x_{t-1},\vec
  x_t|\vec\obs_{1:t-1}]\cov[\vec
  x_t|\vec\obs_{1:t-1}]\inv\nonumber\\
  &=\mat\Sigma_{t-1,t|t-1}^x(\mat\Sigma_{t|t-1}^x)\inv\,.
\end{align}
\end{proposition}

\begin{proof}
  The smoothed state distribution at the terminal time step $T$ is
  equivalent to the filter distribution $\prob(\vec
  x_T|\vec\obs_{1:T})$~\cite{Anderson2005, Bishop2006}. The
  distributions $\prob(\vec x_{t-1}|\vec\obs_{1:T})$, $t =
  T,\dotsc,1$, of the smoothed states can be computed recursively
  according to
\begin{align}
  \prob(\vec x_{t-1}|\vec\obs_{1:T}) &=\int\prob(\vec x_{t-1}|\vec
  x_t,\vec\obs_{1:T})\prob(\vec x_t|\vec\obs_{1:T})\d\vec
  x_t =\int\prob(\vec x_{t-1}|\vec x_t,\vec\obs_{1:t-1})\prob(\vec
  x_t|\vec\obs_{1:T})\d\vec x_t
\label{eq:smoothing recursion}
\end{align}
by integrating out the smoothed hidden state at time step $t$. In
Eq.~\eqref{eq:smoothing recursion}, we exploited that $\vec x_{t-1}$
is conditionally independent of the future measurements $\vec
\obs_{t:T}$ given $\vec x_t$.

To compute the smoothed state distribution in
Eq.~(\ref{eq:smoothing recursion}), we need to multiply a
distribution in $\vec x_t$ with a distribution in $\vec x_{t-1}$ and
integrate over $\vec x_t$. To do so, we follow the steps:
\begin{enumerate}
\item[(a)] Compute the conditional $\prob(\vec x_{t-1}|\vec
  x_t,\vec\obs_{1:t-1})$.
\item[(b)] Formulate $\prob(\vec x_{t-1}|\vec x_t,\vec\obs_{1:T})$ as
  an unnormalized distribution in $\vec x_t$.
\item[(c)] Multiply the new distribution with $\prob(\vec x_t|\vec\obs_{1:T})$.
\item[(d)] Solve the integral in Eq.~\eqref{eq:smoothing recursion}.
\end{enumerate}
We now examine these steps in detail. Assume a known (Gaussian)
smoothed state distribution $\prob(\vec x_t|\vec\obs_{1:T})$.
\begin{enumerate}
\item[(a)]
  Compute a Gaussian approximation to the conditional $\prob(\vec
  x_{t-1}|\vec x_t,\vec\obs_{1:t-1})$. We compute the conditional in
  two steps: First, we compute a Gaussian approximation to the joint  distribution
  $\prob(\vec x_t,\vec x_{t-1}|\vec\obs_{1:t-1})$. Second, we apply
  the rules of computing conditionals to this joint Gaussian.
%
Let us start with a Gaussian approximation
\begin{align}
  \gaussBig{
\begin{bmatrix}\vec\mu_{t-1|t-1}^x\\
      \vec\mu_{t|t-1}^x\end{bmatrix}} {\begin{bmatrix}\mat
      \Sigma_{t-1|t-1}^x & \mat\Sigma_{t-1,t|t-1}^x\\
      (\mat\Sigma_{t-1,t|t-1}^x)\T & \mat
      \Sigma_{t|t-1}^x\end{bmatrix}}
\label{eq:p(x_{t-1},x_t|Y)}
\end{align}
to the joint $\prob(\vec x_{t-1},\vec x_t|\vec\obs_{1:t-1})$ and have
a closer look at its components: A Gaussian approximation of the
filter distribution $\prob(\vec x_{t-1}|\vec\obs_{1:t-1})$ at time
step $t-1$ is known and is the first marginal distribution in
Eq.~(\ref{eq:p(x_{t-1},x_t|Y)}). The second marginal $\gauss{\vec
  \mu_{t|t-1}^x}{\mat \Sigma_{t|t-1}^x}$ is the time update and also
known from filtering.
To fully determine the joint in Eq.~(\ref{eq:p(x_{t-1},x_t|Y)}), we
require the cross-covariance matrix 
\begin{align}
  \mat \Sigma_{t-1,t|t-1}^x & =\iint \vec x_{t-1}f(\vec
  x_{t-1})\T\prob(\vec x_{t-1}|\vec\obs_{1:t-1})\d\vec x_{t-1} 
  - \vec\mu_{t-1|t-1}^x(\vec\mu_{t|t-1}^x)\T\,,
\label{eq:cross-covariance smoothing}
\end{align}
where we used the means $\vec\mu_{t-1|t-1}^x$ and $\vec\mu_{t|t-1}^x$
of the measurement update and the time update, respectively. The
zero-mean independent noise in the system Eq.~(\ref{eq:system
  equation}) does not influence the cross-covariance matrix.  The
cross-covariance matrix in Eq.~(\ref{eq:cross-covariance smoothing})
can be pre-computed during filtering since it does not depend on
future measurements.

This concludes the first step (computation of the joint Gaussian) of
the computation of the desired conditional.

In the second step, we apply the rules of Gaussian conditioning to
obtain the desired conditional distribution $\prob(\vec x_{t-1}|\vec
x_t,\vec\obs_{1:t-1})$. For a shorthand notation, we define
\begin{equation}\label{eq:J-matrix smoothing}
\mat J_{t-1}\coloneqq \mat \Sigma_{t-1,t|t-1}^x(\mat
\Sigma_{t|t-1}^x)\inv \,,
\end{equation}
and obtain a Gaussian approximation $\gaussx{\vec x_{t-1}}{\vec
    m}{\mat S}$ of  the conditional distribution  $\prob(\vec x_{t-1} |\vec
x_t,\vec\obs_{1:t-1})$ with
\begin{align}
  \vec m &= \vec\mu_{t-1|t-1}^x+\mat J_{t-1}(\vec x_t -
  \vec\mu_{t|t-1}^x)\,,\\
  \mat S &= \mat \Sigma_{t-1|t-1}^x - \mat J_{t-1}(\mat
  \Sigma_{t-1,t|t-1}^x)\T\,.
\end{align}

\item[(b)] Formulate $\gaussx{\vec x_{t-1}}{\vec m}{\mat S}$ as an
  unnormalized distribution in $\vec x_t$. The square-root of the
  exponent of $\gaussx{\vec x_{t-1}}{\vec m}{\mat S}$ contains
\begin{align}
  \vec x_{t-1} -\vec m &= \vec r(\vec x_{t-1}) -
  \mat J_{t-1}\vec x_t\nonumber
\end{align}
with $\vec r(\vec x_{t-1})= \vec x_{t-1} - \vec\mu_{t-1|t-1}^x+\mat
J_{t-1}\vec\mu_{t|t-1}^x$, which is a linear function of both $\vec
x_{t-1}$ and $\vec x_t$. We now reformulate the conditional Gaussian
$\gaussx{\vec x_{t-1}}{\vec m}{\mat S}$ as a Gaussian in $\mat J_{t-1}
\vec x_t$ with mean $\vec r(\vec x_{t-1})$ and the unchanged
covariance matrix $\mat S$. We obtain the conditional
\begin{align}
\gaussx{\vec x_{t-1}}{\vec m}{\mat S}
&=c_1\gaussx{\vec x_t}{\vec a}{\mat
    A}\,,\label{eq:N(Rx) is N(x)}\\ \text{with} \quad
  c_1&=\sqrt{|2\pi(\mat J_{t-1}\T\mat S\inv\mat
      J_{t-1})\inv|/|2\pi \mat S|}\,,\nonumber
\end{align}
and $\vec a=\mat J_{t-1}\inv\vec r(\vec x_{t-1})\,, \mat A = (\mat
J_{t-1}\T\mat S\inv\mat J_{t-1})\inv$. Note that $\gaussx{\vec
  x_{t-1}}{\vec m}{\mat S}$ is an unnormalized Gaussian in $\vec x_t$,
see Eq.~\eqref{eq:N(Rx) is N(x)}. The matrix $\mat J_{t-1}$
defined in Eq.~(\ref{eq:J-matrix smoothing}) is quadratic, but not
necessarily invertible, in which case we take the
pseudo-inverse. However, we will see that this inversion will
be unnecessary to obtain the final result.

\item[(c)] Multiply the new distribution with $\prob(\vec
  x_t|\vec\obs_{1:T})$. To determine $\prob(\vec
  x_{t-1}|\vec\obs_{1:T})$, we multiply the Gaussian in
  Eq.~(\ref{eq:N(Rx) is N(x)}) with the smoothed Gaussian state
  distribution 
$\gaussx{\vec
    x_t}{\vec\mu_{t|T}^x}{\mat \Sigma_{t|T}^x}$, which yields the
  Gaussian approximation 
\begin{align}
c_1\gaussx{\vec
    x_t}{\vec a}{\mat A}\gaussx{\vec
    x_t}{\vec\mu_{t|T}^x}{\mat \Sigma_{t|T}^x} 
  &= c_1c_2(\vec a)\gaussx{\vec x_t}{\vec b}{\mat B}
\label{eq:density mult backward}
\end{align}
of  $\prob(\vec x_{t-1},\vec x_t|\vec\obs_{1:T})$,
for some $\vec b$, $\mat B$, where $c_2(\vec a)$ is the inverse
normalization constant of $\gaussx{\vec x_t}{\vec b}{\mat B}$.

\item[(d)] Solve the integral in Eq.~\eqref{eq:smoothing
    recursion}. Since we integrate over $\vec x_t$ in
  Eq.~(\ref{eq:smoothing recursion}), we are solely interested in the
  parts that make Eq.~(\ref{eq:density mult backward}) unnormalized,
  i.e., the constants $c_1$ and $c_2(\vec a)$, which are independent
  of $\vec x_t$. The constant $c_2(\vec a)$ in Eq.~(\ref{eq:density
    mult backward}) can be rewritten as $c_2(\vec x_{t-1})$ by
  reversing the step that inverted the matrix $\mat J_{t-1}$, see
  Eq.~(\ref{eq:N(Rx) is N(x)}). Then, $c_2(\vec x_{t-1})$ is given by
\begin{align}
&\hspace{-2mm}c_2(\vec x_{t-1}) =c_1\inv\gaussx{\vec x_{t-1}}{\vec\mu_{t-1|T}^x}{\mat
\Sigma_{t-1|T}^x}\,,
\label{eq:moments of smoothed distribution}\\
&\hspace{-2mm}\vec\mu_{t-1|T}^x=\vec\mu_{t-1|t-1}^x + \mat J_{t-1}(\vec\mu_{t|T}^x
-\vec\mu_{t|t-1}^x)\,,
\label{eq:mean smoother}\\
&\hspace{-2mm}\mat \Sigma_{t-1|T}^x
=\mat \Sigma_{t-1|t-1}^x + \mat J_{t-1}
(\mat \Sigma_{t|T}^x -
\mat \Sigma_{t|t-1}^x)\mat J_{t-1}\T\,.\label{eq:cov smoother}
\end{align}
Since $c_1c_1\inv=1$ (plug Eq.~\eqref{eq:moments of smoothed
  distribution} into Eq.~\eqref{eq:density mult backward}), the
desired smoothed state distribution is
\begin{equation}\label{eq:smoothing result}
\prob(\vec
x_{t-1}|\vec\obs_{1:T})=\gaussx{\vec
x_{t-1}}{\vec\mu_{t-1|T}^x}{\mat \Sigma_{t-1|T}^x}\,,
\end{equation}
where the mean and the covariance are given in Eq.~(\ref{eq:mean
  smoother}) and Eq.~(\ref{eq:cov smoother}), respectively.
\end{enumerate}
This result concludes the proof of Prop.~\ref{prop:smoothing}.
\end{proof}

\subsection{Sufficient Conditions for Smoothing}
After filtering, to determine a Gaussian approximation to the
distribution $\prob(\vec x_{t-1}|\vec\obs_{1:T})$ of the smoothed
state at time $t-1$, only a few additional ingredients are required:
the matrix $\mat J_{t-1}$ in Eq.~\eqref{eq:J-matrix smoothing} and
Gaussian approximations to the smoothed state distribution $\prob(\vec
x_{t}|\vec\obs_{1:T})$ at time $t$ and the predictive distribution
$\prob(\vec x_t|\vec\obs_{1:t-1})$. Everything but the matrix $\mat
J_{t-1}$ can be precomputed either during filtering or in a previous
step of the smoothing recursion. Note that $\mat J_{t-1}$ can also be
precomputed during filtering.

Hence, for Gaussian RTS smoothing it is sufficient to determine
Gaussian approximations to both the joint distribution $\prob(\vec
x_t, \vec\obs_t|\vec\obs_{1:t-1})$ of the state and the measurement
for the filter step and the joint distribution $\prob(\vec
x_{t-1},\vec x_t|\vec\obs_{1:t-1})$ of two consecutive states.


\section{Implications and Theoretical Results}
\label{sec:results}
%
\label{sec:implications}
Using the results from Secs.~\ref{sec:filtering}
and~\ref{sec:smoothing}, we conclude that for filtering and RTS
smoothing it is sufficient to compute or estimate the means and the
covariances of the joint distribution $\prob(\vec x_{t-1},\vec
x_t|\vec\obs_{1:t-1})$ between two consecutive states (smoothing) and
the joint distribution $\prob(\vec x_t,\vec\obs_t|\vec\obs_{1:t-1})$
between a state and the subsequent measurement (filtering and
smoothing).
%
This result has two implications: 
\begin{enumerate}
\item Gaussian filters\slash smoothers can be distinguished by
  their approximations to these joint distributions.
\item If there exists an algorithm to compute or to estimate the means and
  the covariances of the joint distributions $\prob(\vec x,h(\vec
  x))$, where $h\in\{f,g\}$, the algorithm can be used for filtering
  and RTS smoothing.
\end{enumerate}

In the following, we first consider common filtering and smoothing
algorithms and describe how they compute Gaussian approximations to
the joint distributions $\prob(\vec x_{t-1},\vec
x_t|\vec\obs_{1:t-1})$ and $\prob(\vec
x_t,\vec\obs_t|\vec\obs_{1:t-1})$, respectively, which emphasizes the
first implication (Sec.~\ref{sec:example realizations}).  After that,
for the second implication of our results, we take an algorithm for
estimating means and covariances of joint distributions and turn this
algorithm into a filter/smoother (Sec.~\ref{sec:gibbs filter}).

\subsection{Current Algorithms for Computing the Joint
  Distributions}
\label{sec:example realizations}
\begin{table}
  \caption{Computing the means and the covariances of 
    $p(\vec x_t,\vec\obs_t|\vec\obs_{1:t-1})$ and $p(\vec x_{t-1},\vec
    x_t|\vec\obs_{1:t-1})$.} 
\vspace{2mm}
\label{tab:joint}
\centering
\scalebox{0.95}{
\begin{tabular}{l||l|l|l}
  & Kalman filter/smoother &  EKF/EKS & UKF/URTSS and CKF/CKS$^\star$\\
  \hline
  \hline
  $\hat{\vec\mu}_{t|t-1}^x$ 
  & $\vec F\vec{\mu}_{t-1|t-1}^x$
  & $\tilde{\vec F}\hat{\vec\mu}_{t-1|t-1}^x$
  & $\sum_{i=0}^{2D}w_m^{(i)}f(\vec X_{t-1|t-1}^{(i)})$
  \\
  $\hat{\vec\mu}_{t|t-1}^z$
  & $\vec G\vec\mu_{t|t-1}^x$
  & $\tilde{\vec G}\hat{\vec\mu}_{t|t-1}^x$
  &$\sum_{i=0}^{2D}w_m^{(i)}g(\vec X_{t|t-1}^{(i)})$\\
  \hline
  $\hat{\vec\Sigma}_{t|t-1}^x$ & $\vec F\vec\Sigma_{t-1|t-1}^x\vec F^\top+\vec Q$
  & $\tilde{\vec F}\hat{\vec\Sigma}_{t-1|t-1}^x\tilde{\vec F}^\top+\vec Q$
  & $\sum_{i = 0}^{2D} w_c^{(i)}(f(\vec X_{t-1|t-1}^{(i)}) - \vec\mu_{t|t-1}^x)^2 + \vec Q$\\
  $\hat{\vec\Sigma}_{t|t-1}^z$  &  $\vec G\vec\Sigma_{t|t-1}^x\vec G^\top + \vec R$
  & $\tilde{\vec G}\hat{\vec\Sigma}_{t|t-1}^x\tilde{\vec G}^\top + \vec R$
  & $\sum_{i = 0}^{2D} w_c^{(i)}(g( \vec X_{t|t-1}^{(i)}) - \vec\mu_{t|t-1}^z)^2 + \vec R$\\
  \hline
  $\hat{\vec\Sigma}_{t|t-1}^{xz}$ & $\vec\Sigma_{t|t-1}^x\vec G^\top$  
  & $\hat{\vec\Sigma}_{t|t-1}^x\tilde{\vec G}^\top$
  &$\sum_{i = 0}^{2D} w_c^{(i)}( \vec X_{t|t-1}^{(i)} -  \vec
  \mu_{t|t-1}^x)( g(\vec X_{t|t-1}^{(i)}) 
  -  \vec \mu_{t|t-1}^z)^{\top}$ \\
  $\hat{\vec\Sigma}_{t-1,t|t}^x$ &  $\vec\Sigma_{t-1|t-1}^x\vec F^\top$ &
  $\hat{\vec\Sigma}_{t-1|t-1}^x\tilde{\vec F}^\top$
  & $\sum_{i = 0}^{2D} w_c^{(i)}(\vec X_{t-1|t-1}^{(i)} - 
  \vec\mu_{t-1|t-1}^x)(f(\vec X_{t|t-1}^{(i)}) - \vec\mu_{t|t-1}^x)^\top$
\end{tabular}
}
\end{table}
Tab.~\ref{tab:joint} gives an overview of how the Kalman filter, the
EKF, the UKF, and the CKF represent the means and the
(cross-)covariances of the joint distributions $\prob(\vec
x_t,\vec\obs_t|\vec\obs_{1:t-1})$ and $\prob(\vec x_{t-1},\vec
x_t|\vec\obs_{1:t-1})$. In Tab.~\ref{tab:joint}, we use the shorthand
notation $\vec a^2\coloneqq \vec a\vec a\T$. For example, we defined
$(f(\mat X_{t-1|t-1}^{(i)}) - \vec\mu_{t|t-1}^x)^2\coloneqq(f(\mat
X_{t-1|t-1}^{(i)}) - \vec\mu_{t|t-1}^x) (f(\mat
X_{t-1|t-1}^{(i)})-\vec\mu_{t|t-1}^x)\T$.

In the Kalman filter, the transition function $f$ and the measurement
function are linear and represented by the matrices $\mat F$ and $\mat
G$, respectively. The EKF linearizes $f$ and $g$ resulting in the
matrices $\tilde{\mat F}$ and $\tilde{\mat G}$, respectively. The UKF
computes $2D+1$ sigma points $\mat X$ and uses their mappings through
$f$ and $g$ to compute the desired moments, where $w_m$ and $w_c$ are
the weights used for computing the mean and the covariance,
respectively (see~\cite{Thrun2005}, pp. 65). The CKF computations are
nearly equivalent to the UKF's computations with slight modifications:
First, the CKF only requires $2D$ cubature points $\mat X$. The
cubature points are chosen as the intersection of a $D$-dimensional
unit sphere with the coordinate system.  Thus, the sums run from 1 to
$2D$. Second, the weights $w_c=1/D=w_m$ are all
equal~\cite{Arasaratnam2009}.

Although none of these algorithms computes the joint distributions
$\prob(\vec x_t,\vec\obs_t|\vec\obs_{1:t-1})$ and $\prob(\vec
x_{t-1},\vec x_t|\vec\obs_{1:t-1})$ explicitly, they all do so
implicitly.  Using the means and covariances in Fig.~\ref{tab:joint}
in the filtering and smoothing Eqs.~(\ref{eq:proposition filter
  mean}),~(\ref{eq:proposition filter
  covariance}),~(\ref{eq:proposition mean smoother}),
and~(\ref{eq:proposition cov smoother}), the results from the original
papers~\cite{Kalman1960,Rauch1965, Maybeck1979, Julier2004,
  Sarkka2008,Arasaratnam2009} are recovered. To the best of our
knowledge, Tab.~\ref{tab:joint} is the first presentation of the CKS.

\subsection{Gibbs-Filter and Gibbs-RTS Smoother}
\label{sec:gibbs filter}

We now derive a Gaussian filter and RTS smoother based on Gibbs
sampling~\cite{Geman1984}.  Gibbs sampling is an example of a Markov
Chain Monte Carlo (MCMC) algorithm and often used to infer the
parameters of the distribution of a given data set.  In the context of
filtering and RTS smoothing, we use Gibbs sampling for inferring the
mean and the covariance of the distributions $\prob(\vec x_{t-1},\vec
x_t|\vec\obs_{1:t-1})$ and $\prob(\vec
x_t,\vec\obs_t|\vec\obs_{1:t-1})$, respectively, which is sufficient
for Gaussian filtering and RTS smoothing, see
Sec.~\ref{sec:implications}.

\begin{algorithm}[tb]
\caption{Gibbs-RTSS}
\label{alg:gibbs-rtss}
\begin{algorithmic}[1]
\State \textbf{init:} $\prob(x_0),\mat Q,\mat R,f ,g$
\Comment{initializations}
\For{$t=1$ \textbf{to} $T$}
\Comment{forward sweep (Gibbs-filter)}
\State infer moments of $\prob(\vec x_{t-1},\vec
x_t|\vec\obs_{1:t-1})$ 
\Comment{$\approx$ alg.~\ref{alg:Gibbs}}
\State infer moments of $\prob(\vec
x_t,\vec\obs_t|\vec\obs_{1:t-1})$
\Comment{alg.~\ref{alg:Gibbs}}
\State measure $\vec\obs_t$
\State compute $\vec\mu_{t|t}^x,\mat\Sigma_{t|t}^x,\mat J_{t-1}$
\Comment{Eqs.~(\ref{eq:generic filter mean}),~(\ref{eq:generic filter covariance}),~(\ref{eq:J-matrix smoothing})}
\EndFor
\For{$t = T$ \textbf{to} $1$}
\Comment{backward sweep}
\State compute $\vec\mu_{t-1|T}^x,\mat\Sigma_{t-1|T}^x$
\Comment{Eqs.~(\ref{eq:mean smoother}),~(\ref{eq:cov smoother})}
\EndFor
\end{algorithmic}
\end{algorithm}

Alg.~\ref{alg:gibbs-rtss} details the high-level steps of the
Gibbs-RTSS. 

\begin{figure}[tb]
\centering
\includegraphics[height = 3.02cm]{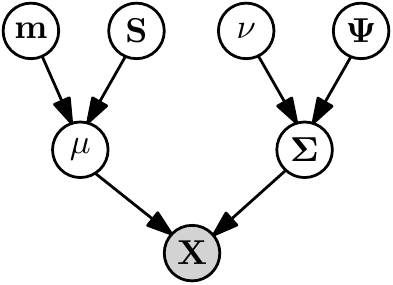}
\label{fig:gibbs filter}
\caption{Graphical model for the Gibbs-filter\slash
  RTSS. $\mat X$ are data from
  a joint distribution, $\vec\mu$ and $\mat\Sigma$ are the mean and
  the covariance of $\mat X$. The parameters of the conjugate
  priors on the mean and the covariance are denoted by $\vec
  m,\mat S$ and $\mat\Psi,\nu$, respectively.}
\end{figure}
At each time step, we use Gibbs sampling to infer the moments of the
joint distributions $\prob(\vec x_{t-1},\vec x_t|\vec\obs_{1:t-1})$
and $\prob(\vec x_t,\vec\obs_t|\vec\obs_{1:t-1})$. Fig.~\ref{fig:gibbs
  filter} shows the graphical model for inferring the mean $\vec\mu$
and the covariance $\mat\Sigma$ from the joint data set $\mat X$ using
Gibbs sampling. The parameters of the conjugate priors on the
mean $\vec\mu$ and the covariance $\mat\Sigma$ are denoted by $\vec
m,\mat S$ and $\mat\Psi,\nu$, respectively.

To infer the moments of the joint $\prob(\vec x_{t-1},\vec
x_t|\vec\obs_{1:t-1})$, we first generate i.i.d. samples from the
filter distribution $\prob(\vec x_{t-1}|\vec\obs_{1:t-1})$ and map
them through the transition function $f$. The samples and their
mappings serve as samples $\mat X$ from the joint distribution
$\prob(\vec x_{t-1},\vec x_t|\vec\obs_{1:t-1})$. With a conjugate
Gaussian prior $\gaussx{\vec\mu}{\vec m}{\mat S}$ on the joint mean,
and a conjugate inverse Wishart prior distribution
$\mathcal{IW}(\mat\Sigma|\mat\Psi,\nu)$ on the joint covariance
matrix, we infer the posterior distributions on $\vec\mu$ and
$\mat\Sigma$.  By sampling from these posterior distributions, we
obtain unbiased estimates of the desired mean and the covariance of
the joint $\prob(\vec x_{t-1},\vec x_t|\vec\obs_{1:t-1})$ as the
sample average (after a burn in).

To infer the mean and the covariance of the joint $\prob(\vec
x_t,\vec\obs_t|\vec\obs_{1:t-1})$, we proceed similarly: We generate
i.i.d. samples from the distribution $\prob(\vec
x_t|\vec\obs_{1:t-1})$, which are subsequently mapped through the
measurement function. The combined data set of i.i.d. samples and
their mappings define the joint data set $\mat X$. Again, we choose a
conjugate Gaussian prior on the mean vector and a conjugate inverse
Wishart prior on the covariance matrix of the joint $\prob(\vec
x_t,\vec\obs_t|\vec\obs_{1:t-1})$. Using Gibbs sampling, we sample
means and covariances from the posteriors and obtain unbiased
estimates for the mean and the covariance of the joint $\prob(\vec
x_t,\vec\obs_t|\vec\obs_{1:t-1})$.

Alg.~\ref{alg:Gibbs} outlines the steps for computing the joint
distribution $\prob(\vec x_t,\vec\obs_t|\vec\obs_{1:t-1})$.
\begin{algorithm}[tb]
  \caption{Inferring the mean $\vec\mu_{x,z}$ and the covariance
    $\mat\Sigma_{x,z}$ of $\prob(\vec
    x_t,\vec\obs_t|\vec\obs_{1:t-1})$ using Gibbs
    sampling}\label{alg:Gibbs}
\begin{algorithmic}[1]
  \State pass in marginal distribution $\prob(\vec
  x_{t}|\vec\obs_{1:t-1})$, burn-in period $B$, number $L$ of Gibbs
  iterations, size $N$ of data set

  \State init. conjugate priors on joint mean and covariance
  $\gaussx{\vec\mu_{x,\obs}}{\vec m}{\mat S}$ and
  $\mathcal{IW}(\mat\Sigma_{x,\obs}|\mat\Psi,\nu)$

  \State $\mat X\coloneqq [\vec x_{t}\idx{i}, g(\vec
  x_{t}\idx{i})+\vec v_t\idx{i}]_{i=1}^N$
  \Comment generate joint data set

  \State sample $\vec\mu_1\sim\gauss{\vec m}{\mat S}$ 

  \State sample $\mat\Sigma_1\sim\mathcal IW(\mat\Psi,\nu)$ 

  \For{$j=1$ \textbf{to} $L$}
  \Comment{for $L$ Gibbs iterations do}

  \State update $\vec m|\mat X,\vec\mu_j,\mat\Sigma_j$ 
  \Comment posterior parameter (mean) of $\prob(\vec\mu_j)$

  \State update $\mat S|\mat X,\vec\mu_j,\mat\Sigma_j$ 
  \Comment posterior parameter (covariance) of
  $\prob(\vec\mu_j)$

  \State sample $\vec\mu_{j+1}\sim\gauss{\vec m}{\mat S}$ 
  \Comment sample mean of the joint

  \State update $\mat \Psi|\mat X,\vec\mu_{j+1},\mat\Sigma_j$ 
  \Comment posterior hyper-parameter (scale matrix) of
  $\prob(\vec\Sigma_{j})$

  \State update $\nu |\mat X,\vec\mu_{j+1},\mat\Sigma_{j}$ 
  \Comment posterior hyper-parameter (degrees of freedom) of
  $\prob(\vec\Sigma_{j})$

  \State sample $\mat\Sigma_{j+1}\sim\mathcal{IW}(\mat \Psi,\nu)$ 
  \Comment sample covariance of the joint
\EndFor

\State $\vec\mu_{x,\obs}\coloneqq \E[\vec\mu_{B+1:L}]$
\Comment{unbiased estimate of the mean of the joint distribution}

\State $\mat\Sigma_{x,\obs}\coloneqq \E[\mat\Sigma_{B+1:L}]$
\Comment unbiased estimate of the covariance of the joint distribution

\State \Return $\vec\mu_{x,\obs},\mat\Sigma_{x,\obs}$
\Comment{return inferred mean and covariance of the joint}
\end{algorithmic}
\end{algorithm}
Since the chosen priors for the mean and the covariance are conjugate
priors, all updates of the posterior hyper-parameters can be computed
analytically~\cite{Gilks1996}.

The moments of $\prob(\vec x_{t-1},\vec x_t|\vec\obs_{1:t-1})$, which
are required for smoothing, are computed similarly by exchanging the
pass-in distributions and the mapping function.

\section{Numerical Evaluation}
\label{sec:numerical evaluation}
As a proof of concept, we show that the Gibbs-RTSS proposed in
Sec.~\ref{sec:gibbs filter} performs well in linear and nonlinear
systems. As performance measures, we consider the expected
root mean square error (RMSE) and the expected negative log-likelihood
(NLL) per data point in the trajectory. 
For a single trajectory, the
NLL is given by
\begin{align}
  \text{NLL} &= -\frac{1}{T+1}\sum_{t=0}^T \log
  \mathcal N(x_t^{\text{truth}}|\mu_{t|\tau}^x,(\sigma_{t|\tau}^x)^2)\,,
\end{align}
where $\tau=t$ for filtering and $\tau=T$ for smoothing.  While the
RMSE solely penalizes the distance of the true state and the mean of
the filtering\slash smoothing distribution, the NLL measures the
coherence of the filtering\slash smoothing distributions, i.e., the
NLL values are high if $x_t^{\text{truth}}$ is an unlikely observation
under $\prob(x_t|\mu_{t|\tau}^x,(\sigma_{t|\tau}^x)^2)$,
$\tau\in\{t,T\}$.  In our experiments, we chose a time horizon $T=50$.

\subsection{Proof of Concept: Linear System}
First, we tested the performance of the Gibbs-filter/RTSS in the
linear system 
\begin{align}
x_{t} &= x_{t-1} + w_t\\
z_t &= -2\,x_t + v_t\,, 
\end{align}
where $w_t\sim\mathcal N(0,1), v_t\sim\mathcal N(0,10), p(x_0) =
\mathcal N(0,5)$. In a linear system, the (E)KF is optimal and
unbiased~\cite{Anderson2005}. The Gibbs-filter/RTSS perform as well as
the EKF/EKS as shown in Fig.~\ref{tab:results linear}, which shows the
expected performances (with the corresponding standard errors) of the
filters/smoothers over 100 independent runs, where
$x_0\sim\prob(x_0)$. The Gibbs-sampler parameters were set to
$(N,L,B)=(1000,200,100)$, Alg.~\ref{alg:Gibbs}.
\begin{figure}[tb]
\centering
\scalebox{1}{
\begin{tabular}{c|cc|cc}
& EKF & Gibbs-filter$^\star$  & EKS & Gibbs-RTSS$^\star$\\
\hline
RMSE 
& $1.11\pm 0.014$ 
& $1.12\pm 0.014$ 
& $0.88\pm 0.011$ 
& $0.89\pm 0.011$ \\ 
NLL &  
$1.52\pm 0.012$ 
& $1.52\pm 0.012$ 
& $1.30\pm0.013$ 
& $1.30\pm 0.012$ 
\end{tabular}
}
\caption{Expected performances (\emph{linear} system) with standard
  error of the mean. The results obtained from the optimal linear
  algorithm and the Gibbs-filter\slash RTSS are nearly identical.}
  \label{tab:results linear} 
\end{figure}

\subsection{Nonlinear System: Non-stationary Growth Model}
As a nonlinear example, we consider the dynamic system
\begin{align}
  x_t &=  \tfrac{x_{t-1}}{2} + \tfrac{25x_{t-1}}{1+x_{t-1}^2} +
  8\cos(1.2\,(t-1)) + w_t\,, \label{eq:nonlinear system}\\
  z_t & = \tfrac{x_t^2}{20} + v_t\label{eq:nonlinear meas}\,,
\end{align}
with exactly the same setup as in~\cite{Doucet2000}: $w_t\sim\mathcal
N(0,1)$, $v_t\sim\mathcal N(0,10)$, and $p(x_0) = \mathcal
N(x_0|0,5)$. This system is challenging for Gaussian filters due to its
quadratic measurement equation and its highly nonlinear system equation.

We run the Gibbs-RTSS, the EKS, the CKS, and the
URTSS~\cite{Sarkka2008} for comparison. We chose the Gibbs parameters
$(N,L,B)=(1000,200,100)$. For 100 independent runs starting from
$x_0\sim\prob(x_0)$, we report the expected RMSE and NLL performance
measures in Fig.~\ref{tab:results nonlinear}.
\begin{figure}[tb]
  \centering
  \scalebox{1}{
\begin{tabular}{c|cccc}
  filters & Gibbs-filter$^\star$ & EKF & CKF & UKF\\
  \hline
  RMSE & $\mathbf{5.04\pm 0.088}$ & $11.1\pm 0.29$ & $6.18\pm 0.17$ & $8.57 \pm 0.16$ \\ 
  NLL & $\mathbf{2.87\pm 0.12}$ &  $26.1 \pm 1.18$ & $9.96\pm 0.75$ & $13.6\pm 0.68$\\
  \hline
  \hline
  smoothers & Gibbs-RTSS$^\star$ & EKS & CKS$^\star$ & URTSS\\
  \hline
  RMSE  & $\mathbf{4.01\pm 0.085}$ & $10.6\pm 0.28$  & $5.66\pm 0.20$ & $8.02\pm 0.16$ \\
  NLL   & $\green{2.78\pm 0.15}$  & $\red{90.6\pm 10.3}$   & $\red{28.9\pm 3.31}$ & $\red{16.3\pm 0.16}$ 
\end{tabular}
}
 \caption{Expected performances (\emph{nonlinear} system) with
   standard error of the mean. The Gibbs-RTSS is the only coherent
   smoother, i.e., it improves the filtering results in the NLL measure.}
\label{tab:results nonlinear} 
\vspace{-4mm}
\end{figure}

Both high expected NLL-values and the fact that smoothing makes them
even higher hint at the incoherencies of the EKF/EKS, the CKF/CKS, and
the UKF/URTSS. The Gibbs-RTSS was the only considered smoother that
consistently improved the results of the filtering step. Therefore, we
conclude that the Gibbs-filter/RTSS is coherent.

Fig.~\ref{fig:results} shows example realizations of filtering and
smoothing using the Gibbs-filter\slash RTSS, the EKF\slash EKS, the
CKF\slash CKS, and the UKF\slash URTSS, respectively.
\begin{figure}[tb]
  \centering 
\subfigure[Gibbs-filter (Gibbs-RTSS). RMSE: 5.56 (4.18),
  \green{NLL: 2.65 (2.45)}.]{
\includegraphics[width = 0.48\hsize]{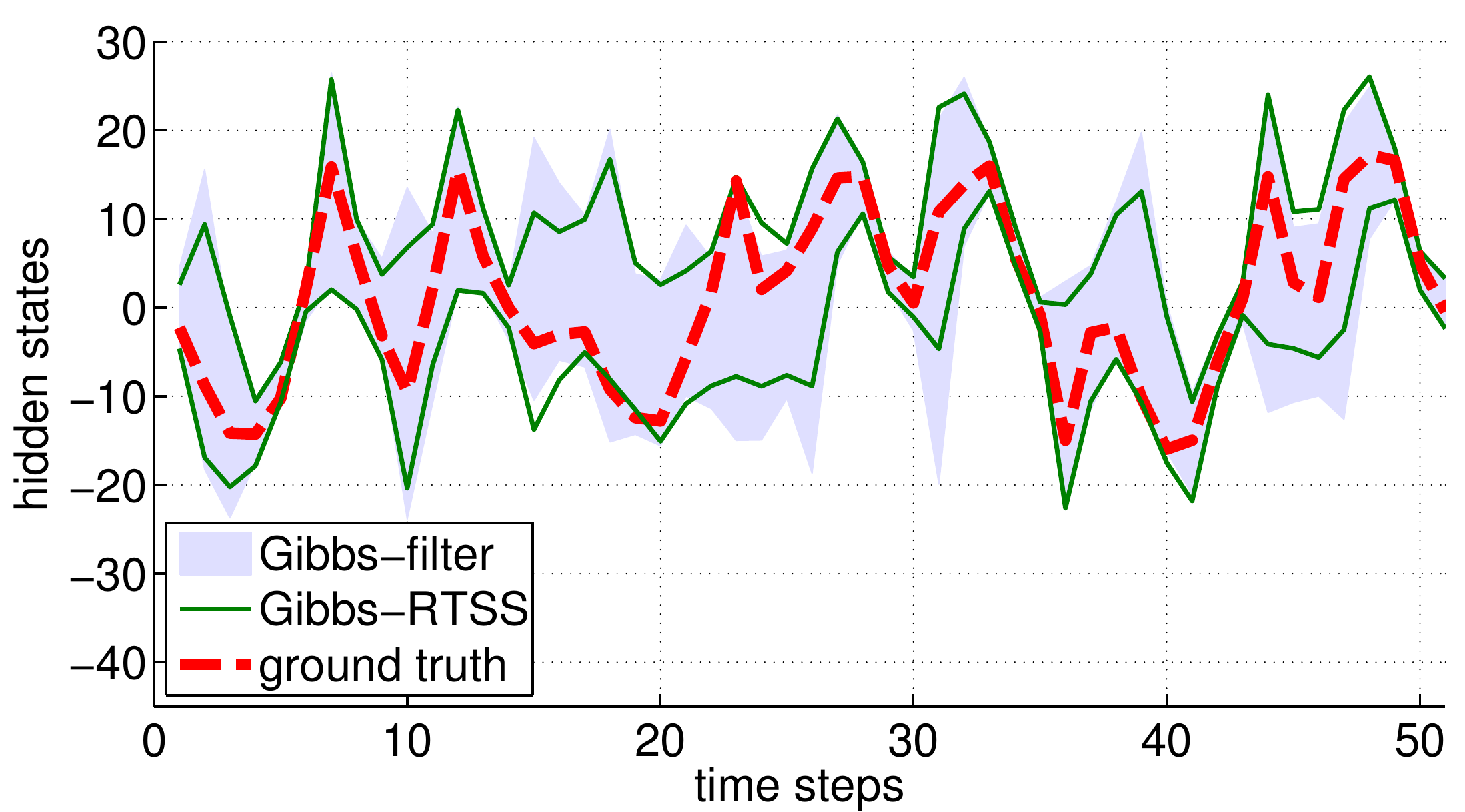}
\label{fig:example gibbs smoother}
}
\hfill
\subfigure[EKF (EKS). RMSE: 11.3 (14.7), \red{NLL: 16.5 (20.8)}.]{
\includegraphics[width = 0.45\hsize]{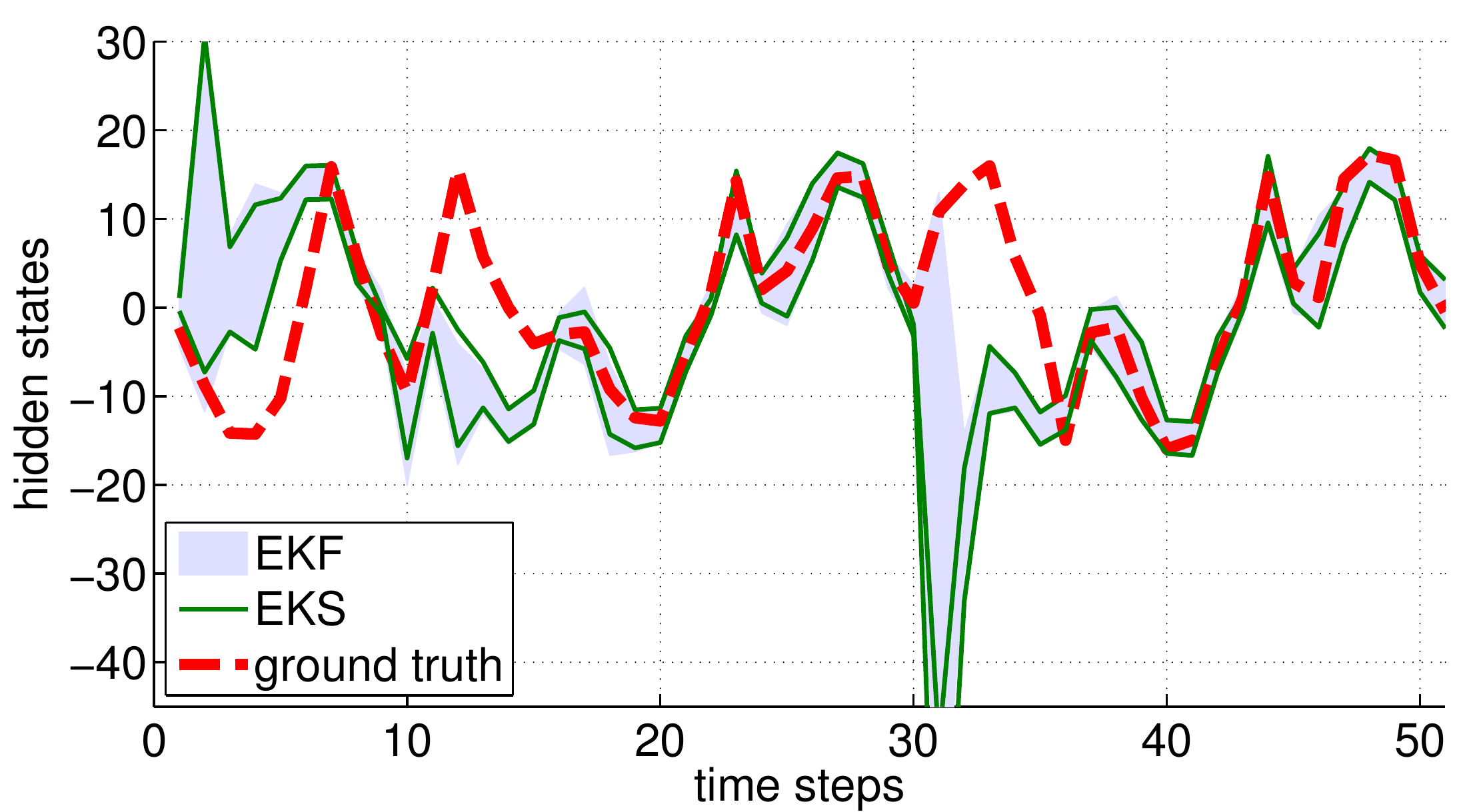}
\label{fig:example eks}
}
\subfigure[CKF (CKS). RMSE: 5.66 (5.96), \red{NLL: 7.32 (20.7)}.]{
\includegraphics[width = 0.48\hsize]{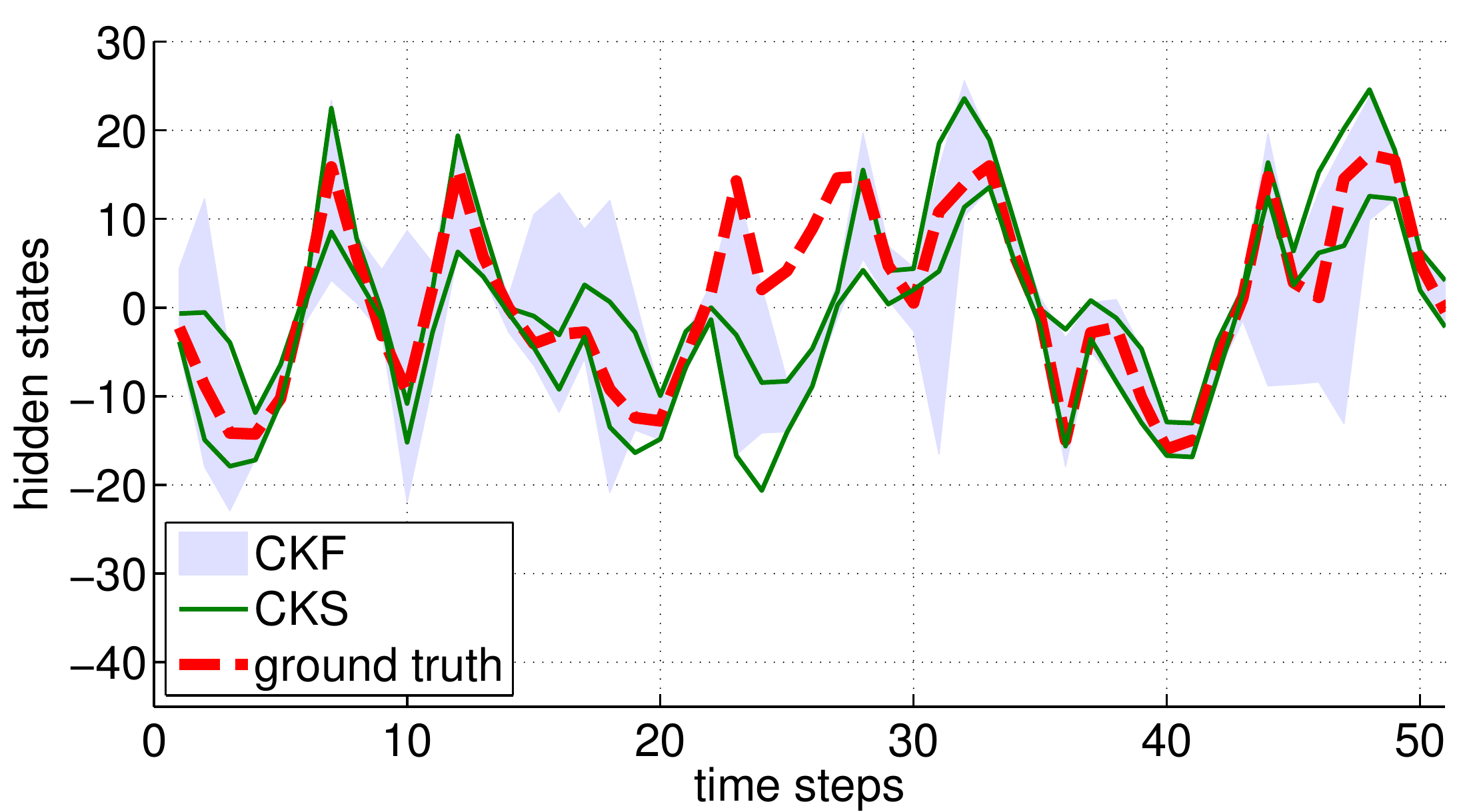}
\label{fig:example cks}
}
\hfill
\subfigure[UKF (URTSS). RMSE: 7.87 (7.18), \red{NLL: 8.66 (9.93)}.]{
\includegraphics[width = 0.45\hsize]{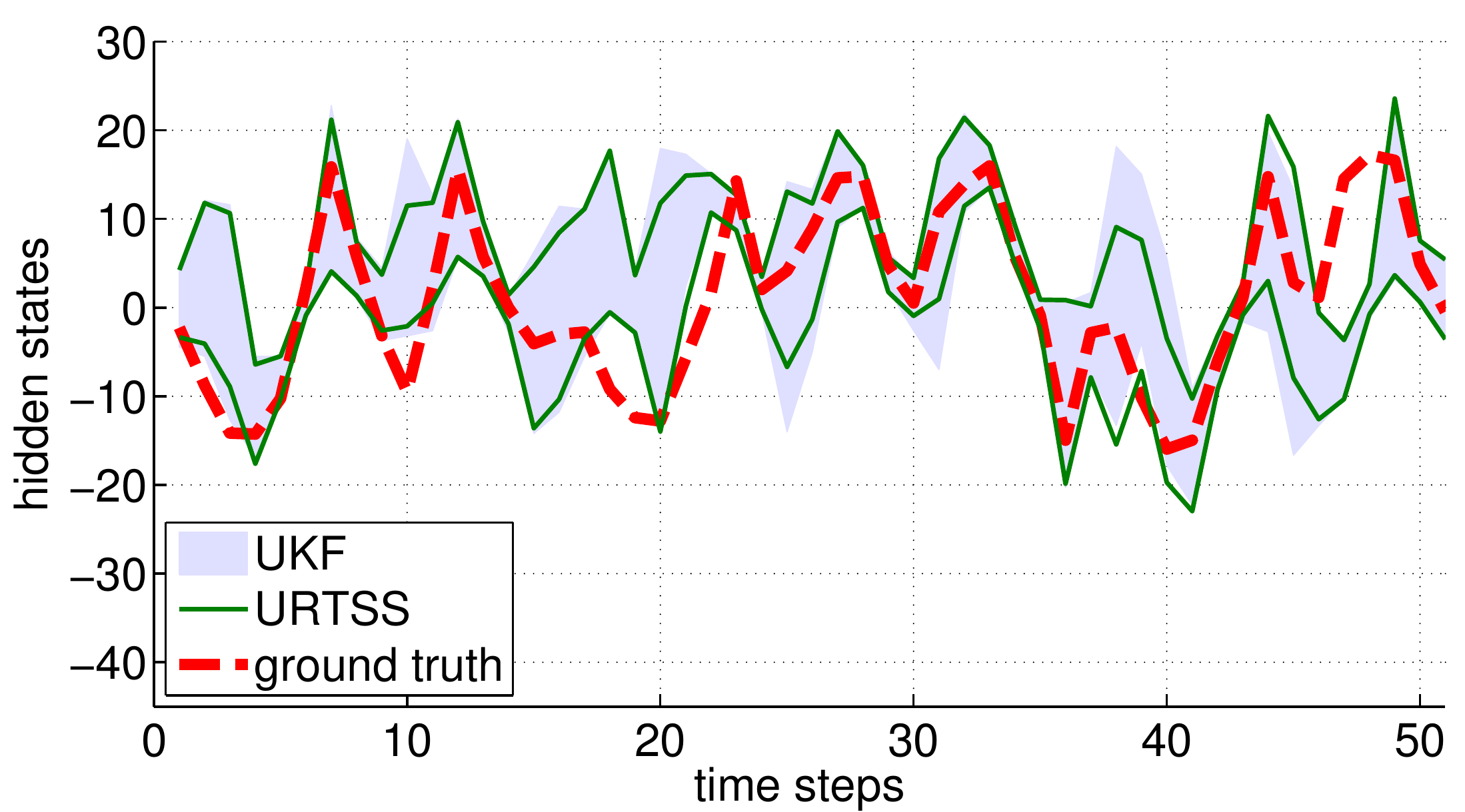}
\label{fig:example urtss}
}
\caption{Example trajectories of filtering\slash smoothing in the
  nonlinear growth model using \subref{fig:example gibbs smoother}
  Gibbs-RTSS, \subref{fig:example eks} EKS, \subref{fig:example cks}
  CKS, \subref{fig:example urtss} URTSS. The filter distributions are
  represented by the shaded areas ($95\%$ confidence area), the
  smoothing distributions are shown by solid green lines ($95\%$
  confidence area). The actual realization of the latent state is the
  dashed red graph.}
\label{fig:results}
\vspace{-4mm}
\end{figure}
The Gibbs-filter/RTSS appropriately inferred the variances of the
latent state while the other filters/smoothers did not (neither of
them is moment-preserving), which can lead to incoherent
filtering/smoothing distributions~\cite{Deisenroth2010b}, see also
Fig.~\ref{tab:results nonlinear}.

\section{Discussion}
\label{sec:discussion}
Our Gibbs-filter/RTSS differs from~\cite{Carter1994}, where Gibbs
sampling is used to infer the noise in a linear system. Instead, we
infer the means and covariances of the full joint distributions
$\prob(\vec x_{t-1},\vec x_t|\vec\obs_{1:t-1})$ and $\prob(\vec
x_t,\vec\obs_t|\vec\obs_{1:t-1})$ in nonlinear systems from
data. Neither the Gibbs-filter nor the Gibbs-RTSS require to know the
noise matrices $\mat R,\mat Q$, but they can be inferred as a part of
the joint distributions if access to the dynamic system is
given. Unlike the Gaussian particle filter~\cite{Kotecha2003}, the
proposed Gibbs-filter is not a particle filter. Therefore, it does not
suffer from degeneracy due to importance
sampling.

Although the Gibbs-filter is computationally more involved than the
EKF/UKF/CKF, it can be used as a baseline method to evaluate the
accuracy and coherence of more efficient algorithms: When using
sufficiently many samples the Gibbs-filter can be considered a close
approximation to a moment-preserving filter in nonlinear stochastic
systems.

The sampling approach to inferring the means and covariances of two
joint distributions proposed in this paper can be extended to infer
the means and covariances of a single joint, namely, $\prob(\vec
x_{t-1},\vec x_t,\vec\obs_t|\vec\obs_{1:t-1})$. This would increase
the dimensionality of the parameters to be inferred, but it would
remove slight inconsistencies that appear in the present approach:
Ideally, the marginals $\prob(\vec x_t|\vec\obs_{1:t-1})$, i.e., the
time update, which can be obtained from both joints $\prob(\vec
x_{t-1},\vec x_t|\vec\obs_{1:t-1})$ and $\prob(\vec
x_t,\vec\obs_t|\vec\obs_{1:t-1})$ are identical. Due to the finite
number of samples, small errors are introduced. In our experiments,
they were small, i.e., the relative difference error was smaller than
$10^{-5}$. Using the joint $\prob(\vec x_{t-1},\vec
x_t,\vec\obs_t|\vec\obs_{1:t-1})$ would avoid this kind of error.

The Gibbs-filter\slash RTSS only need to be able to evaluation the
system and measurement functions. No further requirements such as
differentiability are needed.
A similar procedure for MCMC-based smoothing is applicable when,
instead of Gibbs sampling, slice sampling~\cite{Neal2003} or
elliptical slice sampling~\cite{Murray2010} is used, potentially
combined with GPs that model the functions $f$ and $g$.

The Gibbs-RTSS code is publicly available at
\url{mloss.org}. 

In the context of Gaussian process dynamic systems, the GP-EKF, the
GP-UKF~\cite{Ko2009}, and the GP-ADF~\cite{Deisenroth2009a} can
directly be extended to smoothers using the results from this
paper. The GP-URTSS (smoothing extension of the GP-UKF) and the
GP-RTSS (smoothing extension of the GP-ADF) are presented
in~\cite{Deisenroth2010b}.


\section{Conclusion}
\label{sec:conclusion}

Using a general probabilistic perspective on Gaussian filtering and
smoothing, we first showed that it is sufficient to determine Gaussian
approximations to two joint probability distributions to perform
Gaussian filtering and smoothing. Computational approaches to Gaussian
filtering and Rauch-Tung-Striebel smoothing can be distinguished by
their respective methods used to determining two joint distributions.

Second, our results allow for a straightforward derivation and
implementation of novel Gaussian filtering and smoothing algorithms,
e.g., the cubature Kalman smoother.  Additionally, we presented a
filtering smoothing algorithm based on Gibbs sampling as an
example. Our experimental results show that the proposed
Gibbs-filter/Gibbs-RTSS compares well with state-of-the-art Gaussian
filters and RTS smoothers in terms of robustness and accuracy.

\subsection*{Acknowledgements}
The authors thank S. Mohamed, P. Orbanz, M. Krainin, and D. Fox for
valuable suggestions and discussions. MPD has been supported by ONR
MURI grant N00014-09-1-1052 and by Intel Labs. HO has been partially
supported by the Swedish foundation for strategic research in the
center MOVIII and by the Swedish Research Council in the Linnaeus
center CADICS.


\end{document}